\documentclass[conference]{IEEEtran}

\usepackage{cite}      

\usepackage{graphicx}

\usepackage{psfrag}    

\usepackage{subfigure} 

\usepackage{url}       

\usepackage{amsmath}   
\usepackage{amsfonts}
\usepackage{amssymb}
\usepackage[all]{xy}
\usepackage{textcomp}
\usepackage{eufrak}

\usepackage{array}
\begin{document}

\title{\mbox{\huge {Capacity of a Class of Modulo-Sum Relay Channels}}}

\author{\authorblockN{Marko Aleksic, Peyman Razaghi, and Wei Yu}
\authorblockA{Department of Electrical and Computer Engineering,
University of Toronto, Canada\\
\text {e-mails:\{aleksicm,peyman,weiyu\}@comm.utoronto.ca}} }



%


\maketitle

\begin{abstract}
   This paper characterizes the capacity of a class of modulo additive noise relay channels, in which the relay observes a corrupted version of the noise and has a separate channel to the destination. The capacity is shown to be strictly below the cut-set bound in general and achievable using a quantize-and-forward strategy at the relay. This result confirms a conjecture by Ahlswede and Han about the capacity of channels with rate limited state information at the destination for this particular class of channels.
\end{abstract}


%
\IEEEpeerreviewmaketitle

\section{Introduction}

The relay channel is a fundamental building block in network
information theory. Complete characterization of the relay channel
capacity would be a first step toward finding the capacities of
larger networks. Although the capacity of the general relay channel
is not yet known, the capacities of many specific classes of relay
channels have been found. These special classes include the
degraded, reversely degraded \cite{covergamal}, orthogonal
\cite{zahedi}, semideterministic \cite{aref}, and recently a special
class of deterministic\cite{coverkim} relay channels. All the above
relay channels for which capacities are characterized have one thing
in common: they achieve their respective cut-set bounds. This makes
converses straightforward. Unfortunately it appears that the cut-set
bound cannot be achieved for many practical relay channels. Efforts
to find different bounds, or prove the looseness of the cut-set
bound have proved to be quite difficult. Zhang's partial
converse\cite{zhang} demonstrated the latter; Zahedi \cite{zahedi}
provided some justifications for why the cut-set bound cannot be tight in all cases. \\
 \indent In this paper we find the capacity for a non-trivial class of modulo-sum relay channels. In these channels, the relay observes a correlated version of the noise between the source and the destination, and has a dedicated channel to the destination. We show that the capacity can be strictly below the cut-set bound, and is achievable by a quantize-and-forward strategy \cite[Theorem 6]{covergamal}. The quantize-and-forward strategy was previously only known to achieve the cut-set bound capacity of one class of deterministic relay channels \cite{coverkim}. The modulo-sum relay channel appears to be a first example of a channel where this strategy achieves a capacity strictly below the cut-set bound.\\ \indent The quantize-and-forward strategy was designed for use in channels where the relay has a poor quality channel from the source. In this strategy the relay quantizes its received signal, and transmits the quantized signal to the destination. The destination first decodes the quantized signal from the relay, then uses this signal to help decode the source message. The destination may also use its own received signal to help the decoding of the quantized signal from the relay, because the two signals may be correlated in a general relay channel. This technique is known as Wyner-Ziv coding. Quantize-and-forward is a natural strategy for the modulo channel considered in this paper where the relay observes only the noise. This is because there is no message for the relay to decode; all the relay can do is to describe the noise to the destination.\\
  \indent The converse result contained in this paper crucially depends on two properties of modulo-sum channels. In these channels a uniform distribution on the input alphabet achieves the maximum possible entropy of the output, regardless of the statistics of the additive noise. Further, under a uniform input distribution, the output of a modulo-sum channel is also independent of the additive noise. This has the consequence of simplifying the converse: the side information in Wyner-Ziv coding is not useful since the destination's observation is independent of the relay's output.\\
   \indent A relay channel where the relay only gets to observe some possibly stochastic function of the noise and has a dedicated finite capacity channel to the destination can be viewed as a channel with rate limited state information available to the destination. The capacity result for modulo-sum relay channels coincides with a hypothesis by Ahlswede and Han \cite{ahlswede} about the capacity of channels with rate limited state information to the destination.

\section{A Binary Symmetric Relay Channel}\label{bin}

We begin by deriving the capacity of a particular binary symmetric relay channel. The derivation will be directly applicable to a broader class of modulo-sum relay channels. The simple binary symmetric case is used to distil the essential steps and ideas.\\
\indent Consider the relay channel as shown in Fig. \ref{second}.
Here, the channel input $X$ goes through a binary symmetric channel
(BSC) with crossover probability $p$ to reach $Y$, i.e., $Y=X+Z$
(mod 2) with $Z$ being an i.i.d. $Ber(p)$ random variable. The relay
gets to observe a noisy version of $Z$, namely $Y_{1}=Z+V$, where
$V$ is an i.i.d. $Ber(\delta)$ random variable. The relay also has a
separate BSC to the destination $S=X_{1}+N$, where $N$ is an i.i.d.
$Ber(\epsilon)$ random variable.
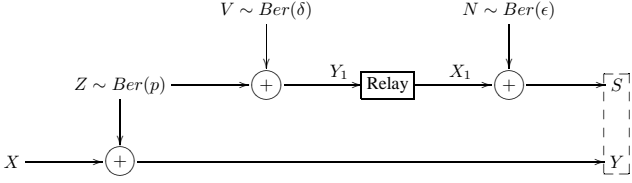
\begin{figure}[h]
\begin{displaymath}
\scalebox{0.68}{\xymatrix{&& V \sim Ber(\delta)\ar[d] && N \sim
Ber(\epsilon)\ar[d]\\
& Z \sim Ber(p)\ar[r]\ar[d]&*++[o][F-]{+}\ar[r]^(.6){\displaystyle
Y_1}&
*+[F]{\txt{Relay}}\ar[r]^(.6){ \displaystyle X_{1} }&*++[o][F-]{+}\ar[r]&{S}\\
X\ar[r]&*++[o][F-]{+}\ar[rrrr]&&&&{Y} \save "2,6"."3,6"*[F--]\frm{}
\restore}}
\end {displaymath}
\caption{A Binary Relay Channel}\label{second}
\end{figure}

 Let us define
\begin{align}
R_{0}=\mathop {\max}\limits_{p(x_{1})} \ I(X_{1};S),\label{eq1}
\end{align}
 for future reference. If there were no corrupting variable $V$, then the capacity of this channel is as recently characterized in \cite{coverkim}
 \begin{align}
 C=\mathop{\max}\limits_{p(x)}\min\{I(X;Y)+ R_{0},I(X;Y,Y_{1})\}.
 \end{align}
 Both hash-and-forward \cite{coverkim}, a strategy where the relay simply hashes $Y_{1}$ into equal sized bins, and the classic quantize-and-forward are capacity achieving. The multiple access cut-set bound is $I(X;Y)+R_{0}$. This bound is obtained by considering the achievable rate assuming that the relay already knew the message the source would like to transmit. One way to interpret the achievability of the multiple access cut-set bound is that if $V$ were absent, decoding $X$ is the same as decoding $Z$. So, the relay, by sending parity information about $Z$, can be interpreted to be performing a version of decode-and-forward, as if it already knows the message; random parities for $Z$ turn into random parities for $X$. This
  interpretation would fail if the relay's observation of $Z$ is corrupted by $V$. To the best of the authors' knowledge, the capacity of this class of relay channels when $V$ is present has not been characterized previously. \\
 \indent The following is a reasonable strategy for this channel. The relay tries to quantize $Y_{1}$ in such a way as to minimize the uncertainty about $Z$ at the destination. The main result of this paper is that the above approach is capacity achieving for a class of modulo-sum relay channels including the channel in Fig. \ref{second}.\\

\newtheorem{theorem}{Theorem}
\begin{theorem}
   The capacity $C$ of the binary relay channel in Fig. \ref{second} is
   \begin{align}
     C = \max_{p(u|y_1): I(U; Y_1) \le R_0} 1 - H(Z|U) \label{cap1}
   \end{align}
   where the maximization may be restricted to $U$'s with $|\mathcal{U}|\leq |\mathcal{Y}_{1}|+2$, and $R_{0}$ is as
    defined in (\ref{eq1}).\\
\end{theorem}

\subsection{Proof of Achievability}

  Fix the input distribution of $X$ as $Ber(\frac{1}{2})$. The capacity can be achieved by a direct application of Theorem $6$ in \cite{covergamal}, if we identify $U$ with $\hat Y_{1}$. A separate proof is provided here for completeness based on the theory of jointly strongly typical sequences \cite{elem}.

\indent  We transmit at rate $R$ over $B-1$ blocks, each of length
$n$. For the last block no message is transmitted. As $B \to
\infty$, $\frac{R(B-1)}{B} $ becomes arbitrarily close to $R$.

 \textit{Codebook Generation:} Generate $2^{nR}$ independently and identically distributed $n$-sequences, $\mathbf{X}(w),w \in \{1 \ldots 2^{nR}\}$ where each element is generated i.i.d.  $\sim \prod_{i=1}^{n}p(x_{i})$, and $p(x_{i})$ has the $Ber(\frac{1}{2})$ distribution. Fix a $p(u|y_{1})$ such that it satisfies the constraint $I(U;Y_{1})\le R_{0}$. Generate $2^{nI(U;Y_{1})}$ i.i.d $n$-sequences,  $\mathbf{U}(t),t \in \{1 \ldots 2^{nI(U;Y_{1})}\}$ where each element is generated i.i.d. $\sim\prod_{i=1}^{n}p(u_{i})$.

 \textit{Encoding:} We describe the encoding for block $i$. To send message $w_{i},w_{i} \in \{1 \ldots 2^{nR}\}$, the transmitter simply sends $\mathbf {X}(w_{i})$. The relay, having observed the entire corrupted noise sequence from the previous block $\mathbf{Y}_{1,i-1}$, looks in its $\mathbf{U}$ codebook and finds a sequence $\mathbf{U}(t_{i})$ that is jointly strongly typical with  $\mathbf{Y}_{1,i-1}$. It encodes and sends its index $t_{i}$ across the private channel to the destination. Only the relay transmits to the destination in the last block B.

 \textit{Decoding:} The destination, upon decoding $t_{i}$, looks for a $w_{i-1}$ such that $\mathbf{X}(w_{i-1})$ is jointly strongly typical with both $\mathbf{U} (t_{i})$, and $\mathbf{Y}_{i-1}$.

 \textit{Analysis of the Probability of Error:} Because of the symmetry of the code construction we can perform the analysis assuming $\mathbf{X}(1)$ was sent over all the blocks. Since the decodings of different blocks are independent we can focus on the probability of error over the first block, and drop the time indices. The error events are:\\

 \begin{tabular}{l l}
$E_{1}:$& $(\mathbf {X}(1),\mathbf{Y},\mathbf{Y}_{1})$ are not jointly strongly typical.\\
$E_{2}:$& $\not \exists t,\ (\mathbf{U}(t),\mathbf{Y}_{1})$ are jointly strongly typical.\\
$E_{3}:$& $(\mathbf {X}(1),\mathbf{Y},\mathbf{U}(t))$ are not jointly strongly typical.\\
$E_{4}:$& The destination makes an error decoding $t$ in the \\ & next block.\\
$E_{5}:$& $\exists w \not = 1, (\mathbf {X}(w),\mathbf{Y},\mathbf{U}(t))$ are jointly strongly \\ & typical.\\
\end{tabular}
 \\

 For $n$ sufficiently large we have $P(E_{1})< \frac{\epsilon}{5B}$, and $P(E_{2}\cap E_{1}^{c})<
 \frac{\epsilon}{5B}$. By the Markov lemma \cite[Lemma 14.8.1]{elem}, since $(\mathbf{X}(1),\mathbf{Y})-\mathbf{Y}_{1}-\mathbf{U}(t)$
  forms a Markov chain, $P(E_{3}\cap E_{1}^{c}\cap E_{2}^{c})<\frac{\epsilon}{5B}$ for $n$ sufficiently large.
 Since by construction $I(U;Y_{1})\leq R_{0}$, the index $t$ can be sent to the destination with an arbitrarily small
 probability of error so $P(E_{4})<\frac{\epsilon}{5B}$. Finally, the probability that another randomly
  generated $\mathbf{X}(w)$ is jointly strongly typical with both $\mathbf{Y}$ and $\mathbf{U}(t)$ is less than $2^{-n(I(X;Y,U)-\gamma)}$. Using the union
  bound, we have, $P(E_{5}\cap \bigcap_{i=1}^{4} E_{i}^{c})< 2^{nR}2^{-n(I(X;Y,U)-\gamma)}$. Thus, when \begin{align}
  R<I(X;Y,U),\end{align} we have $P(E_{5}\cap \bigcap_{i=1}^{4} E_{i}^{c})<\frac{\epsilon}{5B}$, for sufficiently large $n$.
   Now, since $X$ and $U$ are independent, we have \begin{align} I(X;Y,U)&=I(X;Y|U)\\&=H(Y|U)-H(Z|U)\\&= 1 - H(Z|U),\end{align} where $H(Y|U)=1$, because for binary symmetric channels under the uniform
   input distribution $Ber(\frac{1}{2})$, the output $Y$ is independent of the additive noise $Z$, and hence $U$. Collecting terms we see that $P(\bigcup_{i=1}^5 E_{i})<\frac{\epsilon}{B}$, so
   that using the union bound again we can make the probability of error over all of the $B$ blocks less than $\epsilon$ as long as $R<1-H(Z|U)$.

\subsection{Converse}

The converse will be easy once we prove the following lemma.\\

\newtheorem{lemma}{Lemma}
\begin{lemma}
    Let $Z$,~$V$,~$N$ be independent Bernoulli random variables and let $Y_{1}=Z+V$, $Y=X+Z$, and $S=X_{1}+N$ as shown in Fig. \ref{second}. The following inequality
    holds for any encoding scheme at the relay,
   \begin{align}
      H(Z^{n}|S^{n}) \geq \min_{p(u|y_{1}):I(U; Y_1) \le R_0} \
      nH(Z|U) \label{lemma}
   \end{align}
   where the minimization on the right-hand side may be restricted to $U$'s with $|\mathcal{U}|\leq |\mathcal{Y}_{1}|+2$.\\
\end{lemma}


\begin{proof}
The proof of the lemma is closely based on the proof of
\cite[Theorem 14.8.1]{elem}. Fixing an encoding scheme at the relay,
our strategy is to show that there always exists a $U$ for which
$H(Z^{n}|S^{n}) \geq nH(Z|U)$ and $I(Y_{1};U)\leq R_{0} $. This
would allows us to conclude that
\[H(Z^{n}|S^{n}) \geq \min_{p(u|y_{1}):I(U; Y_1) \le R_0} \ nH(Z|U).
\]

We start by  finding a lower bound for $H(Z^{n}|S^{n})$:
\begin{align}
H(Z^{n}|S^{n})&= \sum_{i=1}^{n} H(Z_{i}|S^{n},Z_{1},...,Z_{i-1})\\
        &\geq \sum_{i=1}^{n} H(Z_{i}|S^{n},Z^{i-1},Y_{1}^{i-1})\\
        &= \sum_{i=1}^{n} H(Z_{i}|S^{n},Y_{1}^{i-1})
\end{align}
where in the third line we use the fact that
$Z_{i}-S^{n}Y_{1}^{i-1}-S^{n}Y_{1}^{i-1}Z^{i-1}$ forms a Markov
chain. The Markov chain follows because $Z_{i}$'s are i.i.d.,
$S^{n}$ is only a function of $Y_{1}^{n}$, and $Z_{i}$ can only be
affected by $Z^{i-1}$ through $S^{n}$. Now define
$U_{i}=(S^{n},Y_{1}^{i-1})$, we get:
\begin{align}
        H(Z^{n}|S^{n})\ge \sum_{i=1}^{n} H(Z_{i}|U_{i}).
\end{align}

Next, note that $Z-Y_{1}-X_{1}-S$ forms a Markov chain. As a
result,
\begin{align}
I(X_{1}^{n};S^{n})&\geq I(Y_{1}^{n};S^{n}) \\
        &= \sum_{i=1}^{n} I(Y_{1i};S^{n}|Y_{11},...,Y_{1(i-1)})\\
        &= \sum_{i=1}^{n} I(Y_{1i};S^{n},Y_{1}^{i-1})
\end{align}
where in the third line we use the fact that $Y_{1i}$ is independent of $Y_{1}^{i-1}$ and consequently $I(Y_{1i};Y_{1}^{i-1})=0$. Using our definition of $U$ we get \begin{align}I(X_{1}^{n};S^{n})\geq \sum_{i=1}^{n} I(Y_{1i};U_{i}).\end{align} \\
Recall that $R_{0}=\mathop {\max}\limits_{p(x_{1})} \ I(X_{1};S)$.
Thus, we have shown the following inequalities:
\begin{align}
R_{0}&\geq  \frac{1}{n} \sum_{i=1}^{n} I(Y_{1i};U_{i})\\
      \frac{1}{n}  H(Z^{n}|S^{n})&\geq \frac{1}{n} \sum_{i=1}^{n} H(Z_{i}|U_{i}).
\end{align}
Introducing a standard timesharing random variable $Q$, the above
equations can be rewritten as
\begin{align}
R_{0}&\geq  \frac{1}{n} \sum_{i=1}^{n} I(Y_{1i};U_{i}|Q=i)=I(Y_{1Q};U_{Q}|Q)\\
      \frac{1}{n}  H(Z^{n}|S^{n})&\geq \frac{1}{n} \sum_{i=1}^{n} H(Z_{i}|U_{i},Q=i)= H(Z_{Q}|U_{Q},Q)
\end{align}
Now, since $Q$ is independent of $Y_{1Q}$,  we have
\begin{align}
  I(Y_{1Q};U_{Q}|Q)=I(Y_{1Q};U_{Q},Q)-I(Y_{1Q};Q)=I(Y_{1Q};U_{Q},Q).
\end{align}
Finally, $Y_{1Q}$ and $Z_{Q}$ have the same joint distribution as
$Y_{1}$ and $Z$, so defining $U=(U_{Q},Q)$, $Z=Z_{Q}$ and,
$Y_{1}=Y_{1Q}$, we have shown the existence of a random variable $U$
such that
\begin{align}
R_{0}&\geq I(Y_{1};U)\\
  H(Z^{n}|S^{n})&\geq n H(Z|U)
\end{align}
for any particular encoding scheme at the relay. Since for every
possible encoding scheme at the relay we can construct an i.i.d. $U$
satisfying the above equations, the minimum over all $U$'s
satisfying $I(U;Y)\le R_{0}$ must satisfy (\ref{lemma}). The
cardinality bound is the same as in \cite[Theorem 14.8.1]{elem}.
\end{proof}
\mbox{} \\
\indent The converse can now be proved in a straightforward manner
with:
\begin{align}
   nR &= H(W)   \\[4pt]
      &= I(W;Y^{n},S^{n}) + H(W|Y^{n},S^{n})  \\[4pt]
      &\stackrel{(a)}{\leq}  I(W;Y^{n},S^{n}) + n\epsilon_{n}  \\[4pt]
      &\leq  I(X^{n};Y^{n},S^{n}) + n\epsilon_{n}  \\[4pt]
      &\stackrel{(b)}{=} I(X^{n};Y^{n}|S^{n}) + n\epsilon_{n} \\[4pt]
      &= H(Y^{n}|S^{n}) - H(Y^{n}|S^{n},X^{n}) + n\epsilon_{n} \label{zzw} \\[4pt]
      &\stackrel{(c)}{\leq} n - H(Z^{n}|S^{n},X^{n}) + n\epsilon_{n}  \label{zzy} \\[4pt]
      &=  n - H(Z^{n}|S^{n}) + n\epsilon_{n} \\[4pt]
     &\stackrel{(d)}{\leq} \max_{p(u|y_{1}):I(U; Y_1) \le R_0} \ n(1 - H(Z|U)) + n\epsilon_{n}      \\[4pt]
     &= nC + n\epsilon_{n}
\end{align}
where \vspace{1.5mm}\\
\begin{tabular}{l l}
(a)& follows from Fano's inequality,\\
(b)& follows from the fact that $X^{n}$ is independent of $S^{n}$,\\
(c)& follows from the fact that the maximum entropy\\
 & of a binary random variable of length $n$ is $n$,\\
(d)& follows from Lemma $1$.\vspace{1.5mm}\\
\end{tabular}
Thus, we have shown that for any relaying scheme with a low
probability of error, $R \leq C$.

\subsection{Comments on Theorem 1}

 The capacity of the binary symmetric relay channel considered above is achieved essentially by digitizing the separate channel between the relay and destination. All that matters is that the capacity of the separate channel is sufficiently high to support the relay's description of $U$, the quantization variable. There is no advantage in joint source channel coding at the relay. The input codebook for $X$ is drawn from the uniform $Ber(\frac{1}{2})$ distribution, identical to the capacity achieving distribution if the relay were absent; the source merely increases its rate once the relay is introduced.\\
 \indent There are two conditions which are important for the converse to work. The channel between the source and destination should be additive and modular. These two conditions allow for two crucial simplifications in the converse. First, a uniform input distribution maximizes the output entropy, regardless of any information that the relay may convey about the noise; this was used in (\ref{zzy}). Second, the linear nature of the channel, combined with the expansion in (\ref{zzw}), reduces the role of the relay to essentially source coding with a distortion metric being the conditional entropy of $Z$. This is in contrast to a general relay channel where the relay observes a combination of the source message and noise, so there is an opportunity for the destination to use its received signal to act as side information in the decoding of the relay's quantized message. For the binary symmetric relay channel, the uniform input distribution completely eliminates any aid the destination's output can provide in the decoding of the relay's message; this makes the converse easier to prove.\\

 \subsection{Capacity Can be Below the Cut-set Bound}
To see that the capacity of Theorem 1 can be strictly below the
cut-set bound, consider the case in which $Z^{n}$ has an i.i.d.
$Ber(\frac{1}{2})$ distribution. The capacity can now be evaluated
as \begin{align}C= 1 -
h(h^{-1}(1-R_{0})*\delta),\label{cap2}\end{align} where
$h(p)=-p\log_{2}p - (1-p)\log_{2}(1-p)$, and
$\alpha*\beta=\alpha(1-\beta)+(1-\alpha)\beta$.  This capacity
expression follows by noting that $I(U;Y_{1})=H(Y_{1})-H(Y_{1}|U)$,
so that the constraint in the maximization of Theorem 1 can be
rewritten as \begin{align} H(Y_{1}|U)\geq H(Y_{1})-R_{0}.
\label{zzz}
\end{align} Now we  use Wyner and Ziv's version of the conditional entropy
power inequality for binary random variables \cite{wyner} to claim
that if \begin{align} H(Y_{1}|U) \geq \alpha ,
\label{zzx}\end{align} then
\begin{align} H(Z|U) \geq
h(h^{-1}(\alpha)*\delta), \label{wzin}\end{align} with equality if
$Y_{1}$ given $U$ is a $Ber(h^{-1}(\alpha))$ random variable. Wyner
and Ziv's inequality holds because when $Z$ is $Ber(\frac {1}{2})$
we can write $Z=Y_{1}+V$, where $V$ is $Ber(\delta)$ and $Y_1$ and
$V$ are independent.\\
\indent Now, let $\alpha = H(Y_1)-R_0$. Observe that the $U$ that
achieves equality in (\ref{wzin}), i.e., the $U$ that gives rise to
$Y_1$ given $U$ as $Ber(h^{-1}(H(Y_1)-R_0))$, is precisely the $U$
that minimizes the Hamming distortion of $Y_{1}$ under a rate
constraint $R_0$ in standard rate-distortion theory. This is because
rate-distortion theory states that for binary random variables,
under a rate constraint $R_{0}$, the minimum achievable average
distortion $\nu$ must satisfy $H(\nu)=H(Y_{1}|U)=H(Y_{1})-R_{0}$ and
$Y_1$ given $U$ must be $Ber(\nu)$. Further, as $Y_1$ is
$Ber(\frac{1}{2})$, the distribution of the optimal $U$ is also
$Ber(\frac{1}{2})$. The capacity (\ref{cap2}) follows by using this
$U$ in (\ref{cap1}) and by substituting
$H(Y_{1})=1$ and $\alpha=1-R_0$ in (\ref{wzin}). \\
\indent We now show that the capacity as given in (\ref{cap2}) is
strictly below the cut-set bound. The cut-set bound equals
\cite{covergamal}
\begin{align}
 \max_{p(x,x_{1})}\min \{I(X,X_{1};Y,S),I(X;Y,S,Y_{1}|X_{1})\}. \end{align}
 When $Z$ is $Ber(\frac{1}{2})$, we have
 \begin{align}
 I(X,X_{1};Y,S) &=  H(Y,S) - H(Y,S|X,X_{1})  \\
                &\leq  2 - H(Z,N|X,X_{1})    \label{csb1}\\
                &= 1-H(Z) + 1 - H(N)        \\
                &= R_{0},
 \end{align}
 where the equality in (\ref{csb1}) is achieved by letting $X$ and $X_{1}$
have independent and identical $Ber(\frac{1}{2})$ distributions.

Similarly, for the broadcast bound we have
 \begin{align}
 I(X;Y,S,Y_{1}|X_{1}) &=  I(X;Y|S,Y_{1},X_{1})\\
                &=  H(Y|S,Y_{1},X_{1}) - H(Z|S,Y_{1},X,X_{1}) \\
                &\leq 1 - H(V|S,Y_{1},X,X_{1}) \label{csb2} \\
                &= 1 - H(\delta).
 \end{align}
In the first line, we use the fact that $X$ is independent of
$Y_{1}$ and $S$ given $X_{1}$. In the third line, we again use the
fact that $Y_{1}=Z+V$ and since $Z$ is $Ber(\frac{1}{2})$, so is
$Y_{1}$, thus $Z=Y_{1}+V$, and $Y_1$ and $V$ are independent. The
equality in (\ref{csb2}) is achieved again with $X$ and $X_{1}$ as
independent $Ber(\frac{1}{2})$ distributed random variables. Since
both (\ref{csb1}) and (\ref{csb2}) are achieved with equality with
the same maximizing $p(x,x_1)$, we have shown that the cut-set bound
for this particular channel is equal to
 \begin{align}\min \{R_{0}, 1-H(\delta)\}.\end{align}  The capacity given by (\ref{cap2}) is strictly below the cut-set bound for all values of $R_{0} \geq 1-H(\delta) $.\\


\section{Extension to Modular Relay Channels}

We now extend the capacity results in Section \ref{bin} to include
the general modulo-sum relay channel depicted in Fig. \ref{modch}.
The source and the destination are related by a modulo-sum channel.
The relay observes $Y_{1}$, which is a correlated version of the
noise $Z$ with a conditional distribution $p(y_{1}|z)$. The relay
also has a dedicated channel to the destination with a capacity
\begin{align} R_{0}=\mathop {\max}\limits_{p(x_{1})} \
I(X_{1};S).\label{rnot2}
\end{align}The binary symmetric relay channel considered in Section \ref{bin} is a specific instance of the modulo-sum relay channel. The capacity proof for the binary case can be augmented to give the capacity of the modulo-sum relay channel.\\


\begin{figure}[t]
\begin{displaymath}
\scalebox{0.85}{\xymatrix{
 &Z \ar[r]\ar[d] &*+[F]{\text{$p(y_{1}|z)$}}{+}\ar[r]^{\displaystyle Y_1}&
*+[F]{\txt{Relay}}\ar[r]^{\displaystyle X_{1}}&*+[F]{\txt{$p(s|x_{1})$}}\ar[r]&{S} \\
X\ar[r]&*++[o][F-]{+}\ar[rrrr]&&&&{Y} \save "1,6"."2,6"*[F--]\frm{}
\restore}}
\end {displaymath}
\caption{The modulo-sum Relay Channel}\label{modch}
\end{figure}
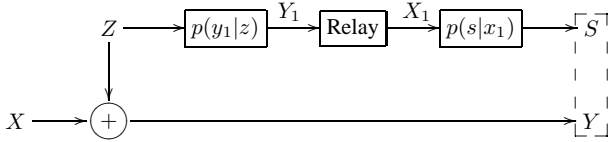


\begin{theorem}
   The capacity of a modular and additive relay channel, in which the relay observes $Y_{1}$, with $p(y_{1}|x,y,z)=p(y_{1}|z)$, and the destination observes $Y=X+Z $ mod $m$ from the source and $S$ from the relay through a separate channel with transition probabilities $p(s|x_{1})$, is
   \begin{align}
      C = \max_{p(u|y_1): I(U; Y_1) \le R_0} m - H(Z|U)
      \label{theorem2}
   \end{align}
   where the maximization may be restricted to $U$'s with $|\mathcal{U}|\leq |\mathcal{Y}_{1}|+2$, and $R_{0}$ is as defined in (\ref{rnot2}).\\

\end{theorem}

Achievability follows by applying a simple extension to the
achievability proof of Theorem $1$. The binary symmetric relay
channel converse appropriately modified to reflect the different
alphabet sizes remains valid. This is because all the necessary
conditions for the converse to work are satisfied. The modulo-sum
channel is linear, and the uniform distribution applied at the input
maximizes the output channel entropy regardless of how much is known
about the additive noise, so (\ref{zzy}) holds.

\section{Connection to Ahlswede-Han conjecture}

The Ahslwede-Han \cite{ahlswede} conjecture states that for channels
with rate limited state information to the decoder as shown in Fig.
\ref{third}, the capacity is given by,
      \begin{align}
      C = \max \ I(X;Y|\hat{S'}) \label{eq2}
   \end{align}
   where the maximum is taken over all probability distributions of the form $p(x)p(s')p(y|x,s')p(\hat{s}'|s')$ such that \[I(\hat{S}';S'|Y)\leq R_{0}\] and the auxillary random variable $\hat {S}'$ has cardinality $|\hat{\mathcal{S}'}|\leq    |\mathcal{S'}|+1$.\\

For these channels, the output $Y$ depends stochastically on both
the input $X$ and the particular channel state $S'$. The channel
state is observed at another encoder that has a digital link to the
destination with capacity $R_{0}$. The conjecture claims that the
state variable $S'$ should be quantized at rate $R_{0}$ in such a
way as to maximize the resulting mutual information between $X$ and
$Y$. By identifying $S'$ with $Y_{1}$, and $\hat{S}'$ with $U$, we
observe that the class of relay channels described in Theorem 2 is a
special case of the channel with rate limited state information to
the decoder. We also note that the uniform distribution on $X$
maximizes the capacity and makes $Y$ independent of $S'$, so that
the rates achievable by (\ref{eq2}) and (\ref{theorem2}) are
identical\footnote{Allowing for the difference in cardinality
bounds.}, thus confirming the conjecture for the class of channels
described in this paper.



\begin{figure}[t]
\begin{displaymath}
\scalebox{0.9}{\xymatrix{&S'\ar[r]\ar[d]& *+[F]{\mathrm{Encoder}}
\ar@2{->}[d]^{
\displaystyle R_{0}} \\
X\ar[r]&*+ [F] {p(y|x,s')}\ar[r]&Y }}\\
\end{displaymath}

\caption{Channel with rate limited state information to the decoder}
\label{third}
\end{figure}
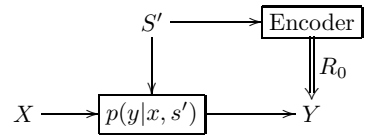

\section{Conclusion}
The capacity of a class of modular additive relay channels was
found. The capacity was shown to be strictly below the cut-set bound
and achievable using a quantize-and-forward scheme where
quantization is performed with a new metric, the conditional entropy
of the noise at the destination. This is the first example of a
relay channel for which the capacity can be strictly below the
cut-set bound. It was proved that there is no advantage to
performing joint source channel coding of the relay's message over
its dedicated link to the destination; digitizing the link is
capacity achieving. The capacity derived here confirms a conjecture
by Ahlswede and Han about the capacity of the rate limited channels
with state information for this class of channels.





\bibliographystyle{IEEEtran}
\bibliography{refs2}

\begin{thebibliography}{1}
\providecommand{\url}[1]{#1}
\csname url@rmstyle\endcsname
\providecommand{\newblock}{\relax}
\providecommand{\bibinfo}[2]{#2}
\providecommand\BIBentrySTDinterwordspacing{\spaceskip=0pt\relax}
\providecommand\BIBentryALTinterwordstretchfactor{4}
\providecommand\BIBentryALTinterwordspacing{\spaceskip=\fontdimen2\font plus
\BIBentryALTinterwordstretchfactor\fontdimen3\font minus
  \fontdimen4\font\relax}
\providecommand\BIBforeignlanguage[2]{{%
\expandafter\ifx\csname l@#1\endcsname\relax
\typeout{** WARNING: IEEEtran.bst: No hyphenation pattern has been}%
\typeout{** loaded for the language `#1'. Using the pattern for}%
\typeout{** the default language instead.}%
\else
\language=\csname l@#1\endcsname
\fi
#2}}

\bibitem{covergamal}
T.~M. Cover and A.~El~Gamal, ``Capacity theorems for the relay channel,''
  \emph{IEEE Trans. Inform. Theory}, vol.~25, no.~5, pp. 572--584, Sept. 1979.

\bibitem{zahedi}
S.~Zahedi, ``On reliable communication over relay channels,'' Ph.D.
  dissertation, Stanford Univ., Stanford, CA, 2005.

\bibitem{aref}
A.~El~Gamal and M.~Aref, ``The capacity of the semideterministic relay
  channel,'' \emph{IEEE Trans. Inform. Theory}, vol.~28, no.~3, p. 536, May
  1982.

\bibitem{coverkim}
T.~M. Cover and Y.-H. Kim, ``Capacity of a class of deterministic relay
  channels,'' in \emph{Proc. IEEE Int. Symp. Information Theory}, June 2007.

\bibitem{zhang}
Z.~Zhang, ``Partial converse for a relay channel,'' \emph{IEEE Trans. Inform.
  Theory}, vol.~34, no.~5, pp. 1106--1110, Sept. 1988.

\bibitem{ahlswede}
R.~Ahlswede and T.~S. Han, ``On souce coding with side information via a
  multiple-access channel and related problems in multi-user information
  theory,'' \emph{IEEE Trans. Inform. Theory}, vol.~29, no.~3, pp. 396--412,
  May 1983.

\bibitem{elem}
T.~M. Cover and J.~S. Thomas, \emph{Elements of Information Theory}.\hskip 1em
  plus 0.5em minus 0.4em\relax New York: Wiley, 1991.

\bibitem{wyner}
A.~D. Wyner and J.~Ziv, ``A theorem on the entropy of certain binary sequences
  and applications: Part {I},'' \emph{IEEE Trans. Inform. Theory}, vol.~19,
  no.~6, pp. 769--772, Nov 1973.

\end{thebibliography}
\end{document}